\documentclass[3p]{elsarticle}

\usepackage{
amsmath,
amsthm,
amscd,
amssymb,
}
\usepackage{wasysym}
\usepackage{comment}
\usepackage{tikz}
\usetikzlibrary{positioning,arrows,decorations.pathreplacing,shapes,calc}

\newtheorem{theorem}{Theorem}
\newtheorem{lemma}[theorem]{Lemma}

\newdefinition{definition}[theorem]{Definition}
\newdefinition{corollary}[theorem]{Corollary}

\DeclareRobustCommand{\P}{\textbf{P}}

\newcommand{\PSPACE}{\textbf{PSPACE}}
\newcommand{\NP}{\textbf{NP}}

\newcommand{\SH}{\mathcal{S}}

\usepackage{framed}

\newcommand{\N}{\mathbb{N}}

\newcommand{\syma}{\textsc{Syn-Majority}}
\newcommand{\sema}{\textsc{Seq-Majority}}
\newcommand{\bsma}{\textsc{Bseq-Majority}}
\newcommand{\clma}{\textsc{Clock-Majority}}
\newcommand{\itci}{\textsc{Iter-Circuit}}
\newcommand{\mitci}{\textsc{Iter-Mon-Circuit}}
\newcommand{\kdmitci}{\textsc{$1$-Depth-Iter-Mon-Circuit}}

\newcommand{\socp}{$\SH$-\textsc{Ocp}}
\newcommand{\bsocp}{\textsc{Bseq-Ocp}}
\newcommand{\syocp}{\textsc{Syn-Ocp}}
\newcommand{\seocp}{\textsc{Seq-Ocp}}
\newcommand{\shmajority}{$\SH$-\textsc{Majority}}

\usepackage[T1]{fontenc}

\title{PSPACE-Completeness of Majority Automata Networks }

\author[uai]{Eric Goles}

\author[uo]{Pedro Montealegre\fnref{becas}}

\author[uc]{Ville Salo\fnref{anillo-basal}}

\author[utu]{Ilkka T\"orm\"a\fnref{aca}}

\address[uai]{Facultad de Ingenier\'ia y Ciencias,  Universidad Adolfo Ib\'a\~nez, Santiago, Chile}

\address[uo]{Univ. Orl\'{e}ans, INSA Centre Val de Loire, LIFO EA 4022,  Orl{\'e}ans , France}

\fntext[becas]{Research supported by Conicyt-Becas Chile-72130083}

\address[uc]{Centro de Modelamiento Matem\'atico, Universidad de Chile, Santiago, Chile}

\fntext[anillo-basal]{Research partially supported by FONDECYT 1140090, ECOS C12E05 and Basal project PFB-03, Centro de
Modelamiento Matem\'atico, UMI 2807 CNRS, Universidad de Chile. }

\address[utu]{University of Turku, TUCS -- Turku Centre for Computer Science}

\fntext[aca]{Research supported by the Academy of Finland Grant 131558}

\begin{document}

\begin{abstract}

We study the dynamics of majority automata networks when the vertices are updated according to a block sequential updating scheme. In particular, we show that the complexity of the problem of predicting an eventual state change in some vertex, given an initial configuration, is \PSPACE-complete.

\end{abstract}


\maketitle

\section{Introduction}

A \emph{threshold network} is a dynamical system over a connected undirected graph, where at each vertex is assigned a \emph{state} that evolves at discrete time steps accordingly to a \emph{vertex threshold function}, that depends on the current state of the vertex and the states of their neighbors in the graph. 
In this paper we study a particular case of threshold networks, called \emph{majority networks}, where at each time step, the vertices take the state that the majority of their neighbors have. 
Such dynamical systems have been used to model a variety of biological, physical and social phenomena \cite{RevModPhys.81.591,PhysRevE.83.056111,Bornholdt06082008,Davidich1}. However, those systems are not always precisely modeled with a synchronous updating of each vertex, which raises the need for considering different ways of updating the network.

An \emph{updating scheme} is a total preorder over the set of vertices, such that at each time step, vertices that are first in this scheme evolve before the others. Updating schemes are classified in three groups: \emph{synchronous}, \emph{sequential} and \emph{block-sequential}. A synchronous updating scheme means that every vertex evolves in parallel. Sequential updating schemes are the other extreme: no two vertices are updated at the same time. The block-sequential updating schemes are an intermediate situation, where the vertices are partitioned into collections called \emph{blocks}, and the vertices of each block update at the same time.

A natural problem in automata networks is \emph{prediction}: given an initial configuration and an updating scheme, to predict the future states. This problem has been studied at least in \cite{TesisGoles,PaperGoles,Hopfield01041982,DBLP:books/daglib/0032430}. One possible strategy is to simulate the evolution of each vertex step by step; since the automata network is finite, the dynamics will eventually enter a loop, and the simulation strategy will result in a complete description of the evolution of the network. A straightforward follow-up question is if this solution is \emph{efficient}, i.e. if there exist better solutions, or if this strategy outputs in \emph{reasonable} time. To evaluate this, the measure chosen is the computational complexity of the problem. In \cite{Goles2014} it is shown that, for threshold networks in general, the simulation strategy is indeed efficient for both synchronous and sequential updating schemes. This was shown by proving that the simulation strategy always stops in a number of steps that is polynomial in the size of the graph, and then the prediction problem is in the class \P. Further, it is shown that for block sequential updating schemes and the majority networks, the problem is unlikely to be efficiently solved (it is \NP-hard), leaving the exact complexity classification open.

In this paper we show that in majority networks the prediction problem, restricted to block-sequential updating schemes, is \PSPACE-Complete. We prove this by showing that majority networks can simulate iterated boolean circuits, whose prediction problem is easily seen to be \PSPACE-Complete. Later, we show that this result remains true even when we limit the block sequential updating schemes to have only blocks of constant size (only a constant number of vertices can be updated at the same time), or a constant number of blocks (the number of groups of vertices that are not updated at the same time is constant). In this context \emph{constant} means independent of the number of vertices of the network. Finally, we show how this result is also generalizable to other decision problems, as well as to a more general form of the majority rule.

\section{Preliminaries}

Let $G = (V,E)$ be a simple connected undirected graph, where $V$ is a finite set of vertices and $E$ the set of edges. An \emph{automata network} is a tuple $\mathcal{A} = (G, (f_v)_{v \in V})$, where $f_v : \{0,1\}^V \rightarrow \{0,1\}$ is the \emph{vertex function} associated to the vertex $v$. Here, $\{0,1\}$ is the set of \emph{states}, and vertices in state $1$ are called \emph{active} while vertices in state $0$ are \emph{inactive}. We say that the vertex functions are the \emph{rule} of $\mathcal{A}$. The set $\{0,1\}^V$ is called the set of \emph{configurations}. We sometimes refer to the sets $V$ or $E$, or the vertex functions $f_v$, without having explicitly introduced them, but in these cases, they should be clear from the context.

Let $N(v) \subset V$ be the set of neighbors of vertex $v \in V$, with $v \notin N(v)$. We are interested in the vertex function, called the \emph{majority function}, defined by
\[ f_v(x) = \left\{ \begin{array}{ll}
 1, & \textrm{if~} \sum_{u \in N(v)} x_j > \frac{\vert N(v) \vert}{2}, \\
 0, & \textrm{if~} \sum_{u \in N(v)} x_j \leq \frac{\vert N(v) \vert}{2}
\end{array}  \right. \]
for all configurations $x \in \{0,1\}^V$. This means that a vertex will become active if the strict majority of its neighbors are active, and will become inactive otherwise. In our proofs, the tie-breaking rule is never activated, so it could be chosen arbitrarily.

An \emph{updating scheme} of the automata network $\mathcal{A}$ is a function $S: V \rightarrow \{1, \dots, \vert V \vert \}$ such that if $u, v \in V$ and $S(u) < S(v)$, then the state of $u$ is updated before $v$, and if $S(u) = S(v)$, then the vertices $u$ and $v$ are updated at the same time. When $S(u) = S(v)$ for all $u, v \in V$, so that all vertices are updated at the same time, we have the \emph{synchronous} updating scheme. When $S$ is injective, that is, no two vertices are updated at the same time, we have a \emph{sequential} updating scheme. In our formalism, every updating scheme $S$ is \emph{block sequential}, meaning that the vertex set is partitioned into the subsets $V_k = S^{-1}(k) \subset V$ for $k \in \{0, \ldots, \vert V \vert\}$, such that the sets are updated one after the other, in the order $V_1, V_2, \ldots, V_{\vert V \vert}$, and elements inside each are updated synchronously.

More precisely, a network $\mathcal{A}$ and an updating scheme $S$ define a \emph{global transition function} $F_S: \{0,1\}^V \rightarrow \{0,1\}^V$ as follows. For each $k \in \{1, \ldots, \vert V \vert\}$, let $F_k : \{0,1\}^V \rightarrow \{0,1\}^V$ be the function that updates the vertices in turn $k$ according to $S$:
\[ F_k(x)_v = \left\{ \begin{array}{ll}
 f_v(x), & \textrm{if~} S(v) = k, \\
 x_v,    & \textrm{otherwise.}
\end{array} \right.
\]
Then we define $F_S = F_n \circ F_{n-1} \circ \dots \circ F_1$. For $t \geq 0$, when the updating scheme $S$ is clear from the context, we sometimes denote $F_S^t(x) = x(t)$ for brevity. The iterated application of the global transition function on a given initial configuration generates a dynamic on the set of all configurations. For instance, Figure \ref{fig:exupdatescheme} shows three updating schemes for a graph of nine vertices and the majority rule. Notice that for the same initial configuration we can obtain different dynamics.

\begin{figure}[htp!]
\centering
\begin{tikzpicture}

\tikzstyle{every node}=[circle,draw, scale=0.85]

\node[draw = none, scale=1/0.85] (x1) at (-2,0) {(a)};
\node[draw = none, scale=1/0.85] (x1) at (-2,-3) {(b)};
\node[draw = none, scale=1/0.85] (x1) at (-2,-6) {(c)};

\node[fill=gray!50] (a) at (0,0) {1};

\node[fill=gray!0] (b) at (0,1) {1};
\draw (a)--(b);
\node[fill=gray!0] (b) at (0.75,0.75) {1};
\draw (a)--(b);
\node[fill=gray!0] (b) at (1,0) {1};
\draw (a)--(b);
\node[fill=gray!0] (b) at (0.75,-0.75) {1};
\draw (a)--(b);
\node[fill=gray!0] (b) at (0,-1) {1};
\draw (a)--(b);
\node[fill=gray!0] (b) at (-0.75,-0.75) {1};
\draw (a)--(b);
\node[fill=gray!0] (b) at (-1,0) {1};
\draw (a)--(b);
\node[fill=gray!0] (b) at (-0.75,0.75) {1};
\draw (a)--(b);

\draw[ultra thick,->] (1.5,0) -- (2.5,0);

\node[fill=gray!0] (a) at (4,0) {1};

\node[fill=gray!50] (b) at (4,1) {1};
\draw (a)--(b);
\node[fill=gray!50] (b) at (4.75,0.75) {1};
\draw (a)--(b);
\node[fill=gray!50] (b) at (5,0) {1};
\draw (a)--(b);
\node[fill=gray!50] (b) at (4.75,-0.75) {1};
\draw (a)--(b);
\node[fill=gray!50] (b) at (4,-1) {1};
\draw (a)--(b);
\node[fill=gray!50] (b) at (3.25,-0.75) {1};
\draw (a)--(b);
\node[fill=gray!50] (b) at (3,0) {1};
\draw (a)--(b);
\node[fill=gray!50] (b) at (3.25,0.75) {1};
\draw (a)--(b);


\node[fill=gray!50] (a) at (0,-3) {9};

\node[fill=gray!0] (b) at (0,-2) {1};
\draw (a)--(b);
\node[fill=gray!0] (b) at (0.75,-2.25) {2};
\draw (a)--(b);
\node[fill=gray!0] (b) at (1,-3) {3};
\draw (a)--(b);
\node[fill=gray!0] (b) at (0.75,-3.75) {4};
\draw (a)--(b);
\node[fill=gray!0] (b) at (0,-4) {5};
\draw (a)--(b);
\node[fill=gray!0] (b) at (-0.75,-3.75) {6};
\draw (a)--(b);
\node[fill=gray!0] (b) at (-1,-3) {7};
\draw (a)--(b);
\node[fill=gray!0] (b) at (-0.75,-2.25) {8};
\draw (a)--(b);

\draw[ultra thick,->] (1.5,-3) -- (2.5,-3);

\node[fill=gray!50] (a) at (4,-3) {9};

\node[fill=gray!50] (b) at (4,-2) {1};
\draw (a)--(b);
\node[fill=gray!0] (b) at (4.75,-2.25) {2};
\draw (a)--(b);
\node[fill=gray!0] (b) at (5,-3) {3};
\draw (a)--(b);
\node[fill=gray!0] (b) at (4.75,-3.75) {4};
\draw (a)--(b);
\node[fill=gray!0] (b) at (4,-4) {5};
\draw (a)--(b);
\node[fill=gray!0] (b) at (3.25,-3.75) {6};
\draw (a)--(b);
\node[fill=gray!0] (b) at (3,-3) {7};
\draw (a)--(b);
\node[fill=gray!0] (b) at (3.25,-2.25) {8};
\draw (a)--(b);

\draw[ultra thick,->] (5.5,-3) -- (6,-3);

\node[draw=none] (p) at (6.6,-3){$\dots$};

\draw[ultra thick,->] (7,-3) -- (7.5,-3);

\node[fill=gray!50] (a) at (9,-3) {9};

\node[fill=gray!50] (b) at (9,-2) {1};
\draw (a)--(b);
\node[fill=gray!50] (b) at (9.75,-2.25) {2};
\draw (a)--(b);
\node[fill=gray!50] (b) at (10,-3) {3};
\draw (a)--(b);
\node[fill=gray!50] (b) at (9.75,-3.75) {4};
\draw (a)--(b);
\node[fill=gray!50] (b) at (9,-4) {5};
\draw (a)--(b);
\node[fill=gray!50] (b) at (8.25,-3.75) {6};
\draw (a)--(b);
\node[fill=gray!50] (b) at (8,-3) {7};
\draw (a)--(b);
\node[fill=gray!50] (b) at (8.25,-2.25) {8};
\draw (a)--(b);


\node[fill=gray!50] (a) at (0,-6) {1};

\node[fill=gray!0] (b) at (0,-5) {2};
\draw (a)--(b);
\node[fill=gray!0] (b) at (0.75,-5.25) {2};
\draw (a)--(b);
\node[fill=gray!0] (b) at (1,-6) {2};
\draw (a)--(b);
\node[fill=gray!0] (b) at (0.75,-6.75) {2};
\draw (a)--(b);
\node[fill=gray!0] (b) at (0,-7) {2};
\draw (a)--(b);
\node[fill=gray!0] (b) at (-0.75,-6.75) {2};
\draw (a)--(b);
\node[fill=gray!0] (b) at (-1,-6) {2};
\draw (a)--(b);
\node[fill=gray!0] (b) at (-0.75,-5.25) {2};
\draw (a)--(b);

\draw[ultra thick,->] (1.5,-6) -- (2.5,-6);

\node[fill=gray!0] (a) at (4,-6) {1};

\node[fill=gray!0] (b) at (4,-5) {2};
\draw (a)--(b);
\node[fill=gray!0] (b) at (4.75,-5.25) {2};
\draw (a)--(b);
\node[fill=gray!0] (b) at (5,-6) {2};
\draw (a)--(b);
\node[fill=gray!0] (b) at (4.75,-6.75) {2};
\draw (a)--(b);
\node[fill=gray!0] (b) at (4,-7) {2};
\draw (a)--(b);
\node[fill=gray!0] (b) at (3.25,-6.75) {2};
\draw (a)--(b);
\node[fill=gray!0] (b) at (3,-6) {2};
\draw (a)--(b);
\node[fill=gray!0] (b) at (3.25,-5.25) {2};
\draw (a)--(b);

\end{tikzpicture}

\caption{Examples of dynamics in a majority network on the same initial configuration under different updating schemes: synchronous (a), sequential (b) and block sequential (c). Gray vertices are active while the white ones are inactive. The number of a vertex is its value in the updating scheme.}
\label{fig:exupdatescheme}
\end{figure}
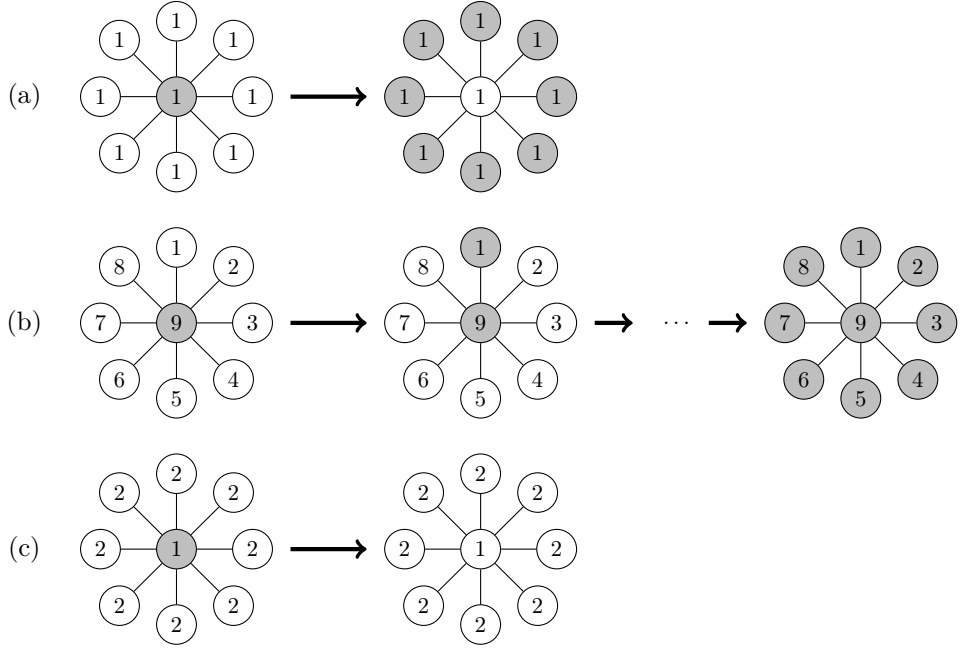

Given an automata network $\mathcal{A}$ and an updating scheme $S$, the \emph{trajectory} of a configuration $x \in \{0,1\}^V$ is the infinite sequence $T_S(x) = (F_S^t(x))_{t \in \N} = (x, x(1), x(2), \ldots)$. When there is no ambiguity, we omit the subindex $S$. The notation $y \in T_S(x)$ means that $y = x(t)$ for some $t \in \N$. Since the set of configurations is finite (its size is $2^{\vert V \vert}$), every trajectory is eventually periodic, so there exist $t \in \N$ and $p > 0$ such that $x(t) = x(t+p)$. We say that the trajectory of $x$ \emph{enters a limit cycle of period $p$} if $p$ satisfies the above and is minimal, and then the set $\{x(t), x(t+1), \dots, x(t+p-1)\} \subset \{0, 1\}^V$ is the \emph{limit cycle} of $x$. A cycle of period $1$ is a \emph{fixed point}. The set of configurations in a limit cycle is called an \emph{attractor}.

We define the \emph{transient length} $\tau_S(x) \in \N$ of a configuration $x$ under the updating scheme $S$ as the number of steps required to reach, for the first time, a configuration in the attractor, and we define the transient length of the automata network $\mathcal{A}$ under the updating scheme $S$ as the the greatest of these values:
\[ \tau_S(\mathcal{A}) = \max \{\tau_S(x) : x \in \{0,1\}^V \}. \]

Suppose that we would like to make predictions about the attractor associated to a configuration $x$, of a majority network $\mathcal{A}$ with updating scheme $S$. Clearly a solution is to simulate the evolution of each vertex until we reach a limit cycle. However, this process can \emph{a priori} take as many steps as there are configurations, or $2^{\vert V \vert}$. A straightforward question is whether more efficient solutions exist, like some algebraic or algorithmic properties that would allow us to make predictions on the size of the attractors, or which vertices will eventually change their initial states. For a family $\SH$ of updating schemes, we define the following decision problem, called the \emph{one vertex prediction problem}: 

\begin{framed}

\noindent {\shmajority}: Given a majority network $\mathcal{A} = (G, (f_v)_{v \in V})$, updating scheme $S \in \SH$, a configuration $x \in \{0,1\}^V$ of $\mathcal{A}$, and a vertex $v \in V(G)$ that is initially inactive ($x_v = 0$), does there exist $y \in T_S(x)$ such that $y_v = 1$?

\end{framed}

The family $\SH$ may contain the synchronous, sequential or block sequential updating schemes. For each case we will have an associated decision problem: \syma, \sema, and \bsma, respectively.


The computational complexity of a decision problem is defined as the amount of resources (usually time or space) required to find an answer, as a function of the input size. Classical complexity theory considers the following fundamental classes: \P{}, the class of problems solvable by a deterministic Turing machine in polynomial time; \NP{}, the class of problems solvable by a nondeterministic Turing machine in polynomial time; and \PSPACE{}, the class of problems solvable by a deterministic Turing machine that uses polynomial space. It is well known that $\P \subset \NP \subset \PSPACE$.\footnote{By $A \subset B$, we denote a not necessarily strict inclusion.} Informally, \P{} is the class of problems with a \emph{feasible} solution in terms of the execution time; \NP{} is the class of problems where it is feasible to verify a given solution; and \PSPACE{} is the class of problems with a feasible solution in terms of the space required to solve the problem.

It is a well-known conjecture that $\P \neq \NP$, and if so, then there exist problems whose solution is feasible to verify, while actually finding a solution is infeasible. The problems in \NP{} which are the most likely to not belong to \P{} are the \NP{}-complete problems, to which any other problem in \NP{} can be reduced by a polynomial time reduction. Thousands of \NP{}-complete problems are known, the best-known example probably being the \emph{Boolean satisfiability problem} \textsc{Sat} \cite{Arora:2009:CCM:1540612}. It is also conjectured that $\NP \neq \PSPACE$, and similarly to \NP{}-Complete problems, the problems in \PSPACE{} that are the most likely to not belong to \NP{} or \P{} are the \PSPACE{}-complete problems, to which any other problem in \PSPACE{} can be reduced in polynomial time. One can also define a notion of completeness for \P{}, and we say a problem is \P{}-complete if every problem in \P{} can be reduced to it by a parallel algorithm in polylogarithmic time.

Since we explicitly use Turing machines with certain properties in our reductions, we give our definition for them. In this article, a \emph{deterministic Turing machine} is a seven-tuple $M = (Q, \Gamma, \Sigma, \delta, B, q_i, q_f)$, containing the \emph{state set} $Q$, \emph{tape alphabet} $\Gamma$, \emph{input alphabet} $\Sigma \subset \Gamma$, \emph{transition function} $\delta : Q \times \Gamma \to Q \times \Gamma \times \{-1, 0, 1\}$, \emph{blank symbol} $B \in \Gamma \setminus \Sigma$, and \emph{initial and final states} $q_i, q_f \in Q$. It operates on a right-infinite tape, and is initialized at the leftmost tape cell. We say that $M$ is \emph{linear bounded}, if there exists $K > 0$ such that on every input $w \in \Sigma^*$, the machine visits at most $K |w|$ distinct tape cells. The \emph{linear bounded prediction problem}, denoted \textsc{Linear-Bounded}, is the problem of determining, given a linear bounded deterministic Turing machine $M$, an input $w \in \Sigma^*$, and a padding of length $K |w|$, whether $M$ accepts $w$ in space $K |w|$. It is easy to show that this problem is \PSPACE-complete.

A \emph{Boolean circuit} is a directed acyclic graph $C$ whose vertices that are not sources are labeled with either $\wedge$ or $\vee$, or possibly $\neg$ if their in-degree is exactly one. The source vertices of $C$ are called its \emph{inputs}, the sinks are called \emph{outputs}, and the other vertices are called \emph{gates}. The circuit is \emph{monotone} is it contains no gates with label $\neg$. If $C$ has $n$ inputs and $m$ outputs, it computes a function $C : \{0,1\}^n \to \{0,1\}^m$ in the obvious way. For each gate of a circuit, its \emph{layer} is the length of the shortest path from an input to the gate. The \emph{iterated (monotone) Boolean circuit problem} is the problem of determining, given a (monotone) Boolean circuit $C : \{0,1\}^n \to \{0,1\}^n$, an input string $x \in \{0,1\}^n$ and an index $i \in \{1, \ldots, n\}$, whether there exists $t \in \N$ such that $C^t(x)_i = 1$. It is denoted \itci{} (\mitci{}, respectively).




\subsection{Previous results}

Threshold automata networks, called also \emph{neural networks}, have been widely studied \cite{TesisGoles,PaperGoles,Hopfield01041982,DBLP:books/daglib/0032430}. In \cite{PaperGoles}, Goles \emph{et al.} give a characterization of the attractors and transients of such networks through the use of a monotone operator, analogous to the \emph{spin glass energy} \cite{Hopfield01041982}. It is shown that for the synchronous updating scheme, when the adjacency matrix $A$ of the underlying graph is symmetric (equivalently, the graph is undirected, which always holds in our case), then the attractors are only fixed points or cycles of period two. Further, if the diagonal elements of the matrix are non-negative (it has no self-loops of negative weight, which again holds in our case), then the sequential updating scheme admits only fixed points. If $A$ is the adjacency matrix of an automata network of $n$ vertices, the transient lengths are bounded by $\mathcal{O}(n^2)$, no matter what updating scheme is used. In terms of complexity, these results imply that the problems \syma{} and \sema{} are in \P, since the dynamics of the majority rule reaches a limit cycle of length at most $2$ in $\mathcal{O}(n^2)$ steps, so that direct simulation leads to a polynomial time algorithm. In the same paper, it was shown that both problems are \P-complete, which means that they are not likely to be efficiently parallelizable \cite{limitsofpara}. 

For more general updating schemes, that is, block-sequential updating schemes, it is shown in \cite{Goles2014} that it is possible to construct majority automata networks with limit cycles of any period, and using these structures in an appropriate way, it is shown that majority automata networks with block-sequential updating schemes admit limit cycles with super-polynomial period in the size of the network. The possibility of large periods suggests that, unlike in the synchronous and sequential cases, it is not possible to have a monotone operator associated to a majority rule for block-sequential updating schemes. Moreover, in \cite{Goles2014} it is shown that \textsc{Bseq-majority} is \NP-hard, and it is conjectured that this problem is in fact \PSPACE-complete.

\subsection{Contributions}

Our main contribution is the proof of the conjecture proposed in \cite{Goles2014}, showing that \textsc{Bseq-Majority} is \PSPACE-complete, and more specifically, our proof shows that this remains true when we restrict the block sequential updating scheme to have blocks of constant size, or a constant number of blocks. The complete proof is given in Section \ref{section:main}.

In Section \ref{setction:problems} we prove that several other decision problems related to majority networks are also \PSPACE-Complete.  Indeed, we show that if we restrict the block-sequential updating scheme to have a constant number of blocks, the corresponding decision problem is still \PSPACE-Complete. On the other hand, in \textsc{Bseq-Majority} one ask for changes in a single vertex, at some step, given an initial configuration. We show that it is still \PSPACE-complete to ask for changes if the initial configuration is not a part of the input, if we ask for changes for infinitely many steps, or if we ask for changes in every vertex of the network at the same time. 

Finally, in Section \ref{setction:proportion} we show that our results apply to a sort of \emph{generalized majority}, where this time a vertex becomes active not if the majority of their neighbors are active, but if a portion $p$ of the number of neighbors, where $p$ is a fixed constant in $(0,1)$. We finish our paper with some conclusions and remarks.

\section{ \bsma{} is \PSPACE-Complete} \label{section:main}

We begin this section by remarking two properties of majority networks. First, Lemma \ref{lem:Network} states that majority networks can simulate monotone Boolean circuits. Then, Lemma \ref{lem:Clock} states that majority networks can exhibit large limit cycles using blocks of size $2$. The gadgets shown in the proofs are essentially found already in \cite{Goles2014}.

\begin{lemma}
\label{lem:Network}

For every monotone Boolean circuit $C : \{0,1\}^n \to \{0,1\}^m$, there is a majority network $\mathcal{M}$ defined over a graph $G$ of size polynomial in $|C|$ with vertices $v_1, \ldots, v_n, w_1, \ldots, w_m \in V(G)$, and with a sequential update scheme $S$ such that for all $I \in \{0,1\}^n$, if we set $x_{v_i} = I_i$ for all $i \in \{1, \ldots, n\}$ and $x_v = 0$ for all other $v \in V(G)$, then $F_S(x)_{w_j} = C(x)_j$ for all $j \in \{1, \ldots, m\}$. Moreover, if every logic gate of $C$ has degree at most $d$, then every node of $G$ has degree at most $2d-1$.

\end{lemma}

\begin{proof}
We can assume that the inputs of $C$ have out-degree exactly one. The circuit $C$ is transformed into $G$ as follows. To the inputs and outputs are assigned the vertices $v_i$ and $w_j$, respectively. The AND-gates and OR-gates are replaced by the gadgets shown in Figure~\ref{fig:Gates}. For the update scheme we can choose any ordering where the inputs of a logic gate are updated before the gate itself, except that the inputs are updated last.
\end{proof}

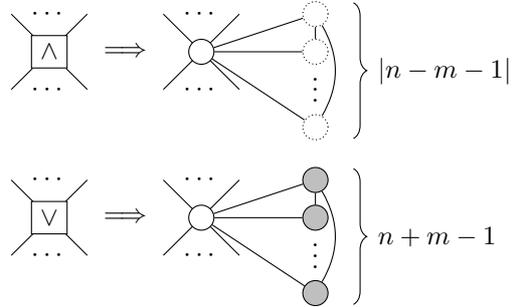
\begin{figure}
\begin{center}
\begin{tikzpicture}

\draw (0,0) -- (-.5,.5);
\draw (0,0) -- (.5,.5);
\draw (0,0) -- (-.5,-.5);
\draw (0,0) -- (.5,-.5);
\node[draw,fill=white] () at (0,0) {$\wedge$};
\node () at (0,.5) {$\cdots$};
\node () at (0,-.5) {$\cdots$};
\node () at (1,0) {$\Longrightarrow$};
\draw (2,0) -- (1.5,.5);
\draw (2,0) -- (2.5,.5);
\draw (2,0) -- (1.5,-.5);
\draw (2,0) -- (2.5,-.5);
\node[draw,fill=white,circle] (g1) at (2,0) {};
\node[draw,densely dotted,fill=white,circle] (v11) at (3.5,.5) {};
\node[draw,densely dotted,fill=white,circle] (v12) at (3.5,0) {};
\node[draw,densely dotted,fill=white,circle] (v13) at (3.5,-1) {};
\path
	(g1) edge node {} (v11) edge node {} (v12) edge node {} (v13)
	(v11) edge node {} (v12)
	(v11) edge [bend left] node {} (v13);
\node () at (3.5,-.4) {$\vdots$};
\node () at (2,.5) {$\cdots$};
\node () at (2,-.5) {$\cdots$};

\draw[decorate,decoration={brace,amplitude=5pt}] (4,.65) -- (4,-1.15);
\node () at (4.2,-.25) [right] {$|n-m-1|$};

\begin{scope}[yshift=-2.2cm]
\draw (0,0) -- (-.5,.5);
\draw (0,0) -- (.5,.5);
\draw (0,0) -- (-.5,-.5);
\draw (0,0) -- (.5,-.5);
\node[draw,fill=white] () at (0,0) {$\vee$};
\node () at (0,.5) {$\cdots$};
\node () at (0,-.5) {$\cdots$};
\node () at (1,0) {$\Longrightarrow$};
\draw (2,0) -- (1.5,.5);
\draw (2,0) -- (2.5,.5);
\draw (2,0) -- (1.5,-.5);
\draw (2,0) -- (2.5,-.5);
\node[draw,fill=white,circle] (g1) at (2,0) {};
\node[draw,fill=gray!50,circle] (v11) at (3.5,.5) {};
\node[draw,fill=gray!50,circle] (v12) at (3.5,0) {};
\node[draw,fill=gray!50,circle] (v13) at (3.5,-1) {};
\path
	(g1) edge node {} (v11) edge node {} (v12) edge node {} (v13)
	(v11) edge node {} (v12)
	(v11) edge [bend left] node {} (v13);
\node () at (3.5,-.4) {$\vdots$};
\node () at (2,.5) {$\cdots$};
\node () at (2,-.5) {$\cdots$};

\draw[decorate,decoration={brace,amplitude=5pt}] (4,.65) -- (4,-1.15);
\node () at (4.2,-.25) [right] {$n+m-1$};
\end{scope}

\end{tikzpicture}
\end{center}
\caption{The gadgets of the logic gates (the gates have indegree $n$ and outdegree $m$). Inputs are on the top. The dotted nodes have value $0$ if $n-m-1 \geq 0$, and $1$ otherwise.}
\label{fig:Gates}
\end{figure}

\begin{lemma}
\label{lem:Clock}
There exists a majority automata network $\mathcal{M}$ over a graph $H$ of degree $3$ containing the vertices $v_s$ for $s \in \{0,1\}^3$ and four additional vertices, a block sequential update scheme $S$, and an initial configuration $x$ such that for all $s \in \{0,1\}^3$ and $t \in \N$ we have $x(t)_{v_s} = s_{t \bmod 3}$.
\end{lemma}

\begin{proof}
This graph is shown in Figure~\ref{fig:TheClock}. The updating scheme is given by the numbering in the figure, and an easy simulation confirms the claim. 
\end{proof}

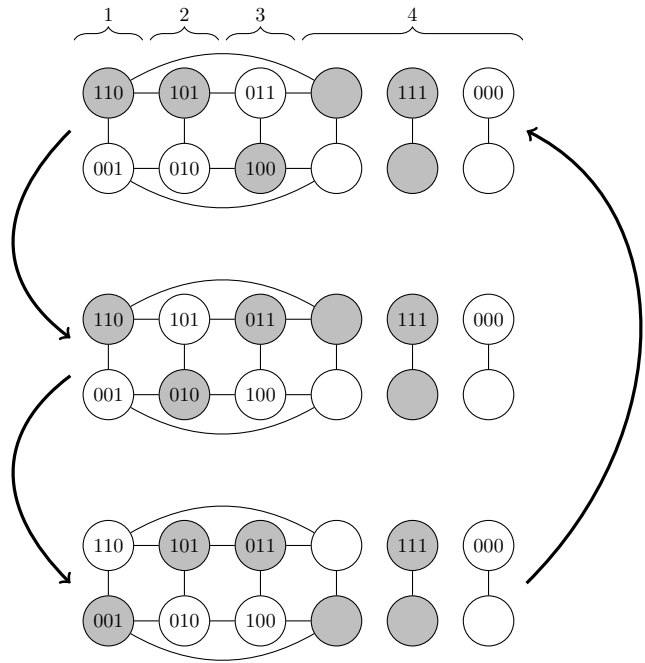
\begin{figure}
\begin{center}
\begin{tikzpicture}[scale=1]

\tikzstyle{every node}=[circle,draw, scale=0.75]

\node (v1) at (0,0) {$001$};
\node[fill=gray!50] (v2) at (0,1) {$110$};
\node (v3) at (1,0) {$010$};
\node[fill=gray!50] (v4) at (1,1) {$101$};
\node[fill=gray!50] (v5) at (2,0) {$100$};
\node (v6) at (2,1) {$011$};
\node (v7) at (3,0) {$\phantom{001}$};
\node[fill=gray!50] (v8) at (3,1) {$\phantom{110}$};
\node[fill=gray!50] (v9) at (4,0) {\phantom{$111$}};
\node[fill=gray!50] (v10) at (4,1) {$111$};
\node (v11) at (5,0) {\phantom{$000$}};
\node (v12) at (5,1) {$000$};

\draw (v1) -- (v3) -- (v5) -- (v7);
\draw (v2) -- (v4) -- (v6) -- (v8);
\draw (v1) -- (v2);
\draw (v3) -- (v4);
\draw (v5) -- (v6);
\draw (v7) -- (v8);
\path (v1) edge [bend right] (v7);
\path (v2) edge [bend left] (v8);
\draw (v9) -- (v10);
\draw (v11) -- (v12);
\draw[decorate,decoration={brace,amplitude=5pt}] (-.45,1.7) -- (.45,1.7);
\draw[decorate,decoration={brace,amplitude=5pt}] (.55,1.7) -- (1.45,1.7);
\draw[decorate,decoration={brace,amplitude=5pt}] (1.55,1.7) -- (2.45,1.7);
\draw[decorate,decoration={brace,amplitude=5pt}] (2.55,1.7) -- (5.45,1.7);
\node[draw=none] () at (0,1.8) [above] {$1$};
\node[draw=none] () at (1,1.8) [above] {$2$};
\node[draw=none] () at (2,1.8) [above] {$3$};
\node[draw=none] () at (4,1.8) [above] {$4$};

\node (v1) at (0,-3) {$001$};
\node[fill=gray!50] (v2) at (0,-2) {$110$};
\node[fill=gray!50] (v3) at (1,-3) {$010$};
\node (v4) at (1,-2) {$101$};
\node (v5) at (2,-3) {$100$};
\node[fill=gray!50] (v6) at (2,-2) {$011$};
\node (v7) at (3,-3) {$\phantom{001}$};
\node[fill=gray!50] (v8) at (3,-2) {$\phantom{110}$};
\node[fill=gray!50] (v9) at (4,-3) {\phantom{$111$}};
\node[fill=gray!50] (v10) at (4,-2) {$111$};
\node (v11) at (5,-3) {\phantom{$000$}};
\node (v12) at (5,-2) {$000$};

\draw (v1) -- (v3) -- (v5) -- (v7);
\draw (v2) -- (v4) -- (v6) -- (v8);
\draw (v1) -- (v2);
\draw (v3) -- (v4);
\draw (v5) -- (v6);
\draw (v7) -- (v8);
\path (v1) edge [bend right] (v7);
\path (v2) edge [bend left] (v8);
\draw (v9) -- (v10);
\draw (v11) -- (v12);

\node[fill=gray!50] (v1) at (0,-6) {$001$};
\node (v2) at (0,-5) {$110$};
\node (v3) at (1,-6) {$010$};
\node[fill=gray!50] (v4) at (1,-5) {$101$};
\node (v5) at (2,-6) {$100$};
\node[fill=gray!50] (v6) at (2,-5) {$011$};
\node[fill=gray!50] (v7) at (3,-6) {$\phantom{001}$};
\node (v8) at (3,-5) {$\phantom{110}$};
\node[fill=gray!50] (v9) at (4,-6) {\phantom{$111$}};
\node[fill=gray!50] (v10) at (4,-5) {$111$};
\node (v11) at (5,-6) {\phantom{$000$}};
\node (v12) at (5,-5) {$000$};

\draw (v1) -- (v3) -- (v5) -- (v7);
\draw (v2) -- (v4) -- (v6) -- (v8);
\draw (v1) -- (v2);
\draw (v3) -- (v4);
\draw (v5) -- (v6);
\draw (v7) -- (v8);
\path (v1) edge [bend right] (v7);
\path (v2) edge [bend left] (v8);
\draw (v9) -- (v10);
\draw (v11) -- (v12);

\draw[->, very thick] (-0.5,0.5) .. controls (-1.5,-0.5) and (-1.5,-1.5) ..  (-0.5,-2.25);
\draw[->, very thick] (-0.5,-2.75) .. controls (-1.5,-3.5) and (-1.5,-4.5) ..  (-0.5,-5.5);
\draw[->, very thick] (5.5,-5.5) .. controls (7.5,-3.5) and (7.5,-0.5) ..  (5.5,0.5);

\end{tikzpicture}
\end{center}
\caption{The clock gadget and its dynamics. Gray vertices are active while the white ones are inactive. The updating scheme is given by the numbering in the top.}
\label{fig:TheClock}
\end{figure}

\begin{lemma}
The problem \mitci{} is \PSPACE-complete, even when restricted to circuits of degree $4$ (in-degree and out-degree $2$).
\end{lemma}

\begin{proof}
First, it is clear that \itci{} is \PSPACE-complete, since the \PSPACE-complete problem \textsc{Linear-Bounded} easily reduces to it. Namely, given a linear bounded Turing machine $M = (Q, \Gamma, \Sigma, \delta, B, q_i, q_f)$, an input $w$ of length $\ell$ and a padding of length $K |w|$, we construct a Boolean circuit with $K \ell \log_2 \vert \Gamma \vert$ inputs for the tape cells, $\log_2 \vert Q \vert$ inputs for the internal state, $\log_2 K \ell$ inputs for the location of the read-write head on the tape, one extra input for signaling the halting state, and some intermediate gates that calculate the next computation step.

To simulate a Boolean circuit by a monotone one, we replace each input $v$ by two inputs $v_+$ and $v_-$, so that $v_+$ is true if and only if $v$ is true, and $v_-$ is true if and only if $v$ is false; similarly, we replace each gate by two gates, a gate computing the value of the original gate, and another gate computing the negation. A $\neg$-gate is simulated by swapping the two inputs, and we use De Morgan's laws to obtain the substitutions for the monotone gates:
\[ u \wedge v \mapsto (u_+ \wedge v_+, u_- \vee v_-), \quad u \vee v \mapsto (u_+ \vee v_+, u_- \wedge v_-). \]
It is clear that the $v_+$-gates of the monotone circuit behave exactly the gates of the original one. Finally, using the construction in \cite[Theorem~6.2.3]{limitsofpara}, we can transform the monotone circuit into one where the in-degree and out-degree of every gate is bounded by $2$.
\end{proof}

Now we define a more powerful model of majority dynamics on Boolean graphs, and show that the related problem is \PSPACE-complete.

\begin{definition}
A \emph{clocked automata network} is an automata network graph $G$, whose every node $v \in V(G)$ is associated with an element $c(v) \in \{U, 1, 0\}^3$, called its \emph{clock}. The elements $U$, $1$ and $0$ are just labels that refer to `update normally,' `switch to $1$,' and `switch to $0$,' respectively. The dynamics of a clocked Boolean graph is the same as that of the underlying automata network, except that on a timestep $t \in \N$, if the symbol $c(v)_{t \bmod 3}$ of the clock is $1$ or $0$, instead of taking the majority function of its neighborhood, the vertex $v$ simply assumes the respective state.

The decision problem associated to clocked Boolean graphs with majority dynamics and block-sequential updates is denoted \clma{}.
\end{definition}



\begin{lemma}
\label{lemma:clmaPSPACE}
The problem \clma{} is \PSPACE-Complete, even when restricted to graphs of degree $7$ with a sequential update scheme.
\end{lemma}

\begin{proof}
We reduce \mitci{} to the clocked majority problem. See Figure~\ref{fig:ClockedLayers}. Updating is done in the order of the numbering on the left, and within a level, the ordering of nodes can be arbitrary, except that the Boolean circuit $C$ (represented by the ellipse and transformed into a Boolean graph with $k$ update blocks using Lemma~\ref{lem:Network}) is updated in its natural order, that is, inputs of a gate are updated before the gate. We can assume that $C$ has degree $4$, so that the clocked network has degree at most $7$ by Lemma~\ref{lem:Network} and the construction.

The small gray vertices have the clock $(1, 1, 1)$, so they will always be active. Other vertices on levels $1$ and $4+k$ have clock $(0, 0, U)$, the nodes on levels $3, 4, \ldots, 4+k-1$ have clock $(U, 0, 0)$ and those on level $2$ have clock $(0, U, 0)$. It is easy to check that in three steps, if we begin with the top and bottom rows containing an input vector $x \in \{0,1\}^n$ and the other rows being empty, we return to the same configuration except with the top and bottom rows containing the string $C(x)$.
\end{proof}

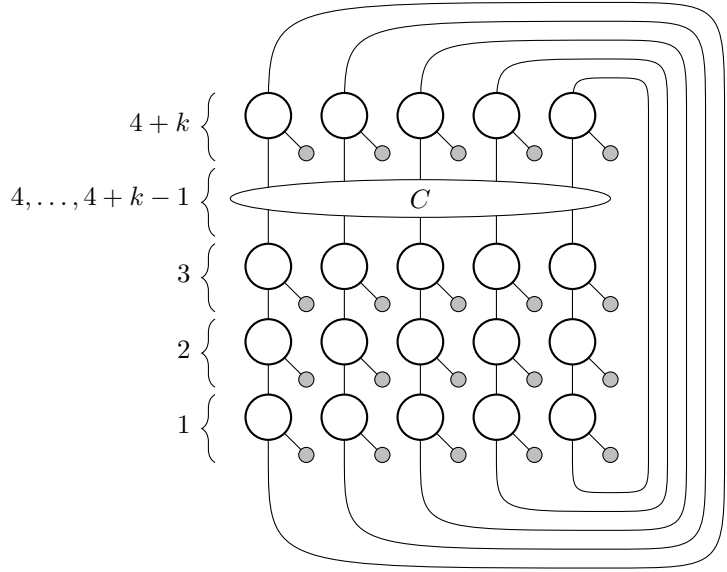
\begin{figure}
\begin{center}
\begin{tikzpicture}

\tikzstyle{vertex}=[circle,thick,draw,fill=white,minimum size=6mm]

\foreach \x in {0,1,...,4}{
	\draw (\x,-.5) -- (\x,2);
	\draw (\x,4) .. controls (\x,5.5-\x/4) .. (4.5,5.5-\x/4) .. controls (6-\x/4,5.5-\x/4) .. (6-\x/4,4) -- (6-\x/4,-.5) .. controls (6-\x/4,\x/4-2) .. (4.5,\x/4-2) .. controls (\x,\x/4-2) .. (\x,-.5);
	\draw (\x,4) -- (\x,2);
	\foreach \y in {0,1,2,4}{
		\draw (\x,\y) -- (\x+.5,\y-.5);
		\node[vertex] () at (\x,\y) {};
		\draw[fill=gray!50] (\x+.5,\y-.5) circle (0.1);
	}
}

\draw[fill=white] (2,2.9) ellipse (2.5 and 0.25);
\node () at (2,2.9) {$C$};

\draw[decorate,decoration={brace,amplitude=5pt}] (-.7,-.6) -- (-.7,.3);
\draw[decorate,decoration={brace,amplitude=5pt}] (-.7,.4) -- (-.7,1.3);
\draw[decorate,decoration={brace,amplitude=5pt}] (-.7,1.4) -- (-.7,2.3);
\draw[decorate,decoration={brace,amplitude=5pt}] (-.7,2.4) -- (-.7,3.3);
\draw[decorate,decoration={brace,amplitude=5pt}] (-.7,3.4) -- (-.7,4.3);

\node () at (-.9,-.1) [left] {$1$};
\node () at (-.9,.9) [left] {$2$};
\node () at (-.9,1.9) [left] {$3$};
\node () at (-.9,2.9) [left] {$4, \ldots, 4+k-1$};
\node () at (-.9,3.9) [left] {$4+k$};

\end{tikzpicture}
\end{center}
\caption{The clocked network, shown here with input length $5$. The small gray vertices will always be active, and the circuit $C$ computes from top to bottom. The number of update blocks in the circuit is $k$. Note how the upper edges reappear cyclically at the bottom.}
\label{fig:ClockedLayers}
\end{figure}

We show now that majority networks can simulate clocked ones. First, we prove the following lemma, which states that using a linear blow-up in the number of vertices, we can make the dynamics of a majority network robust to adding a constant number of vertices to it.

\begin{lemma}
Let $\mathcal{A} = (G, (f_v)_{v \in V(G)})$ be an automata network with a block sequential update scheme $S$, where every the degree of every vertex in $G$ is odd and bounded by $d \in \N$, and let $k \in \N$. Then there exists an automata network $\mathcal{A}^k = (G^k, (\tilde{f}_v)_{v \in V(G^k)})$ of size $(2k + 1)|V(G)|$ and degree at most $(2k + 1)d$ with a block sequential update scheme $S^k$, and a $(2k+1)$-to-one function $\phi : V(G^k) \to V(G)$ such that the following holds. For all $x \in \{0,1\}^{V(G)}$, let $y \in \{0,1\}^{V(G^k)}$ be such that $y_v = x_{\phi(v)}$ for all $v \in V(G^k)$. Then for all $v \in V(G^k)$ and all $t \geq 0$, we have $F_S^t(x)_{\phi(v)} = F_{S^k}^t(y)_v$, even if $G^k$ is an induced subgraph of any graph $H$ of an automata network $\mathcal{H}$ with $|N_H(v) \setminus G^k| \leq k$ for all $v \in V(G^k)$.
\end{lemma}

The automata network $\mathcal{A}^k$ is called the \emph{$k$-amplification} of $\mathcal{A}$.

\begin{proof}
The amplification $\mathcal{A}^k$ has vertex set $V(G^k) = V(G) \times \{1, 2, \ldots, 2k+1\}$ and edge set
\[ E(G^k) = \{ \{(v,i), (w,j)\} \;|\; \{v, w\} \in E(G), i,j \in \{1,2,\ldots,2k+1\} \}. \]
The function $\phi$ and the updating scheme $S$ are simply defined by $\phi((v,i)) = v$ and $S^k((v,i)) = S(v)$ for all $(v,i) \in V(G^k)$. The claim easily follows.
\end{proof}

See Figure~\ref{fig:Amplify} for an example of the amplification operation.

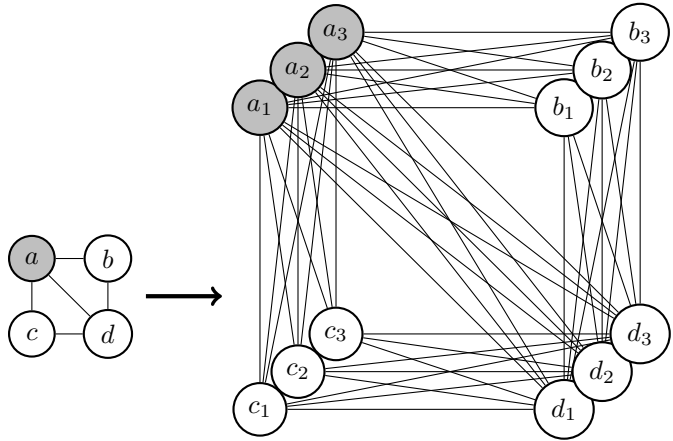
\begin{figure}
\begin{center}
\begin{tikzpicture}

\tikzstyle{every node}=[circle,thick,draw,fill=white,minimum size=6mm]

\node (v1) at (0,0) {$c$};
\node (v2) at (1,0) {$d$};
\node[fill=gray!50] (v3) at (0,1) {$a$};
\node (v4) at (1,1) {$b$};
\draw (v2) -- (v1) -- (v3) -- (v4) -- (v2) -- (v3);

\draw[ultra thick,->] (1.5,.5) -- (2.5,.5);

\foreach \i in {1,2,3}{
	\node (v1\i) at (2.5+\i/2,\i/2-1.5) {$c_\i$};
	\node (v2\i) at (6.5+\i/2,\i/2-1.5) {$d_\i$};
	\node[fill=gray!50] (v3\i) at (2.5+\i/2,\i/2+2.5) {$a_\i$};
	\node (v4\i) at (6.5+\i/2,\i/2+2.5) {$b_\i$};
}
\foreach \i in {1,2,3}{
	\foreach \j in {1,2,3}{
		\draw (v2\i) -- (v1\j) -- (v3\i) -- (v4\j) -- (v2\i) -- (v3\j);
	}
}

\end{tikzpicture}
\end{center}
\caption{The amplification operation applied to a simple majority network with $k = 1$.}
\label{fig:Amplify}
\end{figure}



Let $\mathcal{BS}_k$ be the family of block sequential updating schemes with block size smaller than $k$. 

\begin{theorem}
\label{theo:main}
$\mathcal{BS}_8$-\textsc{Majority} is \PSPACE-complete, even when restricted to graphs of degree at most $23$.
\end{theorem}

\begin{proof}
We reduce \clma{} to said problem. Given a clocked majority automata network $\mathcal{A}$ with underlying graph $G$ of degree $d \in \N$, we construct a majority automata network $\mathcal{B}$ from it as follows. First, we attach a $\lceil d_v/2 \rceil$-fold amplified copy $\mathcal{M}^{\lceil d_v/2 \rceil}$ of the clock network $\mathcal{M}$ of Lemma~\ref{lem:Clock} to every vertex $v \in V(G)$, where $d_v \in \N$ is the degree of $v$. Denote this subgraph by $\mathcal{C}_v$. Now, $\mathcal{C}_v$ contains at least $d_v+1$ copies of every vertex $v_s$ in $\mathcal{M}$, any of which we can attach to $v$ without affecting the behavior of $\mathcal{C}_v$. Suppose that the clock of $v$ is $c(v) \in \{0, 1, U\}^3$, and let $c^\downarrow(v)$ and $c^\uparrow(v)$ be the elements of $\{0, 1\}^3$ obtained by replacing every $U$ in $c(v)$ by $0$ or $1$, respectively. Then we attach $d_v+1$ copies of the vertices $c^\downarrow(v)$ and $c^\uparrow(v)$ in $\mathcal{C}_v$ to $v$; if $c^\downarrow(v) = c^\uparrow(v)$, we attach them only once. For example, if the clock of $v$ is $(1, U, 0)$, then we can attach it to $d_v+1$ copies of the vertex $100$, and $d_v+1$ copies of $110$. It is clear from the properties of the clock gadget that the resulting majority automata network $\mathcal{B}$ behaves like the clocked majority automata network $\mathcal{A}$. Furthermore, if the degree of $G$ is bounded by $d$ and its update scheme is sequential, and if we update the clocks one at a time, then the graph of $\mathcal{B}$ has degree at most $3d+2$ and block size at most $d+1$. The claim follows, since we can choose $d = 7$.
\end{proof}


Since $\mathcal{BS}_k$ is a subset of the whole family of block sequential updating schemes, we obtain as a corollary that \bsma{} is \PSPACE-Complete.

\begin{corollary}
\bsma{} is \PSPACE-complete.
\end{corollary}

We have shown that the problem \bsma{} is \PSPACE-complete when the maximum block size is constant with respect to the size of the majority network. In the following, we will show that if we limit the total number of blocks instead of their size, then the problem is still \PSPACE-Complete.


We start by showing that in our circuit simulations, it is enough to deal with circuits of depth $1$, that is, circuits with just two layers, input and output. Consider the following decision problem.

\begin{framed}
\noindent \kdmitci: Given a depth-$1$ monotone Boolean circuit $C : \{0,1\}^n \to \{0,1\}^n$, an input string $x \in \{0,1\}^n$ and an index $i \in \{1, \ldots, n\}$, decide if whether there exists $t \in \N$ such that $C^t(x)_i = 1$. 
\end{framed}

We show that this problem is \PSPACE-complete. 

\begin{lemma}
\kdmitci{} is \PSPACE-complete, even when restricted to circuits of degree $4$.
\label{lem:kdepth}
\end{lemma}

\begin{proof}

This problem is clearly in \PSPACE. We will show that we can reduce \mitci{} to \kdmitci{} in polynomial time, in a way that preserves the maximum in-degrees and out-degrees of vertices.
Let $(C,x,i)$ an instance of \mitci{}, where $C : \{0 ,1\}^n \to \{0, 1\}^n$ is a circuit of size $N$ and depth $D$, $x \in \{0, 1\}^n$ is an input string and $i \in \{1, \ldots, n\}$ an index. We can assume that $C$ is \emph{synchronous}, that is, every input of a gate in layer $\ell$ comes from layer $\ell-1$, and also that each layer of $C$ has the same number $n = N/D$ of gates. Indeed, we can obtain a synchronous monotone circuit from any monotone circuit by increasing the number of fates by a polynomial factor \cite{limitsofpara}. For each $\ell \in \{0, \ldots, D\}$, we enumerate the gates in layer $\ell$ as $\{ \ell n + 1, \ell n + 2, \ldots, (\ell+1) n \}$.

We build from $C$ a monotone boolean circuit $\overline{C}$ with $D n$ inputs, $D n$ outputs and depth $1$ as follows. First, we enumerate both the input and output gates of $\overline{C}$ as $\{ 1, \ldots, D n \}$, the latter of which have the same types ($\wedge$ or $\vee$) as their counterparts in $C$. Then, for all $\ell \in \{0, \ldots, D-1\}$ and $k, k' \in \{1, \ldots, n\}$, if there is a wire in $C$ from the gate numbered $\ell n + k$ to the gate numbered $(\ell+1) n + k'$, then we add a wire in $\overline{C}$ from the input numbered $\ell n + k$ to the output numbered $(\ell+1) n + k'$.

\begin{figure}[htbp]
 \begin{center}
\begin{tikzpicture}

\tikzstyle{every node}=[circle,draw,scale=0.95]

\node[draw,circle] (a1) at (0,2) {1};
\node[draw,circle] (a2) at (1,2) {2};
\node[draw,circle] (a3) at (2,2) {3};

\node[draw,circle] (b1) at (0,1) {4};
\node[draw,circle] (b2) at (1,1) {5};
\node[draw,circle] (b3) at (2,1) {6};

\node[draw,circle] (c1) at (0,0) {7};
\node[draw,circle] (c2) at (1,0) {8};
\node[draw,circle] (c3) at (2,0) {9};

\draw[->] (a1)--(b1);
\draw[->] (a1)--(b3);
\draw[->] (a2)--(b2);
\draw[->] (a3)--(b3);
\draw[->] (a3)--(b1);

\draw[->] (b1)--(c2);
\draw[->] (b1)--(c3);
\draw[->] (b2)--(c1);
\draw[->] (b2)--(c3);
\draw[->] (b3)--(c1);
\draw[->] (b3)--(c2);

\draw[ultra thick,->] (2.5,1) -- (3,1);

\foreach \i in {1,...,6}{
	\node (d\i) at (2.5+\i,2) {$\i$};
}
\foreach \i in {1,2,3}{
	\node (e\i) at (5.5+\i,0) {$\i$};
}
\foreach \i in {4,5,6}{
	\node (e\i) at (-.5+\i,0) {$\i$};
}

\draw[->] (d1) -- (e4);
\draw[->] (d1) -- (e6);
\draw[->] (d2) -- (e5);
\draw[->] (d3) -- (e6);
\draw[->] (d3) -- (e4);
\draw[->] (d4) -- (e2);
\draw[->] (d4) -- (e3);
\draw[->] (d5) -- (e1);
\draw[->] (d5) -- (e3);
\draw[->] (d6) -- (e1);
\draw[->] (d6) -- (e2);

\end{tikzpicture}
 \end{center}
\caption{An example of transforming a circuit $C$ of depth $2$ into $\overline{C}$. Note that the outputs of $\overline{C}$ are not in the correct order in the figure.}
\end{figure}
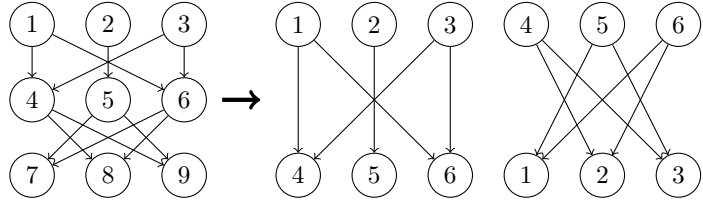 

Let $(x,0^{(D-1)n})$ be the vector of size $D n$ where the first $n$ components are $x$ and the rest are zeros, and let $C_\ell$ be the circuit of $n$ inputs and $n$ outputs and depth $\ell$, obtained by taking the first $\ell$ layers of $C$. It follows directly from the construction of $\overline{C}$ that $\overline{C}^k (x, 0^{(D-1)n}) = (0^{kn}, C^k(x), 0^{(D-(k+1))n})$ for every $1 \leq k \leq D-1$. Moreover, $\overline{C}^D (x, 0^{(D-1)n}) = (C(x),0^{(D-1)n})$. It follows that for all $t \geq 0$, we have $\overline{C}^{t D}(x,0^{(D-1)n})_i = C^t(x)_i$, and $\overline{C}^{t D + \ell}(x,0^{(D-1)n})_i = 0$ if $0 < \ell < D$. Clearly both $\overline{C}$ and $x0^{(D-1)n}$ can be constructed in a time that is polynomial in the size of $C$, and the maximum in-degree and out-degree of $\overline{C}$ are exactly those of $C$, which finishes the proof.

\end{proof}

Let $\mathcal{BN}_k$ the family of block sequential updating schemes with a number of blocks smaller than $k$. 

\begin{corollary}
The problem $\mathcal{BN}_8-\textsc{Majority}$ is \PSPACE-complete, even when restricting to graphs of degree at most $23$.
\end{corollary}

\begin{proof}
This result is proved similarly to Theorem \ref{theo:main}, with the following modifications. First, we can assume by Lemma~\ref{lem:kdepth} that the circuit $C$ has depth $1$, so that it adds no extra levels to the network. Second, the clock gadgets are not updated one by one, but in parallel, as are the vertices of every level. Then the total number of update blocks is the number of levels, plus $4$ (from the clocks), or $8$ in total.
\end{proof}

\section{Other decision problems: the complexity of the Majority Rule}

\label{setction:problems} 

As the title of this paper suggest, we propose that the `correct' measure of the complexity of a class of majority automata networks is the complexity of the prediction problem \shmajority{} restricted to that class, or in other words, the problem of predicting the state changes in some particular vertex of a given network. In this section we give some evidence for this statement by showing that there are many other natural prediction problems that \bsma{} can be reduced to, showing that they are also \PSPACE-complete. Below are three examples of such problems.

\begin{framed}

\noindent {$\SH$-\textsc{Eventual-Prediction}}: Given a majority network $\mathcal{A} = (G, (f_v)_{v \in V})$, an updating scheme $S \in \SH$, an initial configuration $x \in \{0,1\}^V$ of $\mathcal{A}$, and an initially inactive vertex $v \in V(G)$, does there exists $t_0 > 0$ such that $x(t)_v = 1$ for all $t \geq t_0$?

\end{framed}

\begin{theorem}
\textsc{Bseq-Eventual-Prediction} is \PSPACE-complete.
\label{thm:Eventual}
\end{theorem}

\begin{proof}


The \bsma{} problem can be reduced to this problem by attaching the gadget of Figure \ref{fig:EventualGadget} to the vertex $v \in V(G)$ of the \bsma{} problem (the central vertex in the figure), and considering either of the new inactive vertices. They become active as soon as $v$ does, and then stay active forever. The problem is still in \PSPACE, since we can just simulate the graph until we reach a period, which can be easily detected.
\end{proof}

\begin{figure}[htp!]
\begin{center}
\begin{tikzpicture}
\node[draw,circle] (a) at (0,0) {$v$};
\node[draw,circle,fill=gray!50] (b) at (1,0) {};
\node[draw,circle,fill=gray!50] (c) at (1,1) {};

\node[draw,circle,fill=gray!50] (d) at (.5,1.5) {};
\node[draw,circle,fill=gray!50] (e) at (0,1) {};

\node[draw,circle] (f) at (-1,0) {};
\node[draw,circle] (g) at (-1,1) {};

\path
	(a) edge (b) edge (c) edge (f) edge (g)
	(e) edge (f) edge (g) edge (b) edge (c) edge (d)
	(c) edge (b) edge (d)
	(g) edge (f);
\end{tikzpicture}
\end{center}
\label{fig:EventualGadget}
\caption{The gadget of Theorem~\ref{thm:Eventual}.} 
\end{figure}
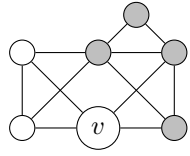 

\begin{framed}

\noindent {$\SH$-\textsc{Conditional-Prediction}}: Given a majority network $\mathcal{A} = (G, (f_v)_{v \in V})$, an updating scheme $S \in \SH$, a subset $W \subset V$ and a partial configuration $y \in \{0,1\}^W$, and an initially inactive vertex $v \in W$, does there exist a configuration $x \in \{0,1\}^V$ with $x|_W = y$ such that $x(t)_v = 1$ for some $t > 0$?


\end{framed}

\begin{theorem}

\textsc{Bseq-Conditional-Prediction} is \PSPACE-Complete.

\end{theorem}

\begin{proof}
This problem is clearly at least as hard as \bsma{}, and it is in \PSPACE{} since we can simply enumerate all possibilities for the values of the vertices $V \setminus W$ in polynomial space, and run the simulations.
\end{proof}

\begin{framed}

\noindent {$\SH$-\textsc{Full-Prediction}}: Given a majority network $\mathcal{A} = (G, (f_v)_{v \in V})$, an updating scheme $S \in \SH$, and an initial configuration $x \in \{0,1\}^V$, does there exist a time $t > 0$ such that $x(t)_v = 1$ for all $v \in V$?

\end{framed}

\begin{theorem}
\label{thm:Full}
\textsc{Bseq-Full-Prediction} is \PSPACE-Complete.

\end{theorem}

\begin{proof}
This problem is also clearly in \PSPACE, and the \bsma{} problem can be reduced to it as follows. Suppose we are given a majority network $\mathcal{A}$ with underlying graph $G$ with maximum degree bounded by an odd number $d \geq 3$, and the special vertex $v \in V(G)$. Enumerate the vertices of $G$ as $v_1, \ldots, v_k$ such that $v_1 = v$. For each $i \in \{1, \ldots, k\}$, add to $G$ two new complete graphs, $K_i^0$ with inactive vertices and $K_i^1$ with active vertices, of respective sizes $d$ and $3d-1$. Introduce an edge from each vertex of $K_i^0$ to $v_i$, to each vertex of $K_{i+1 \bmod k}^0$, and to each vertex in a subset of $K_i^1$ of size $2d$ (in the case $i = 1$, to every vertex of $K_i^1$). From any $d$ vertices of $K_i^1$, introduce new edges to $v_i$. See Figure~\ref{fig:FullGadget} for a visualization of the construction. Call this new majority network $\mathcal{B}$ with underlying graph $H$. We extend the update scheme of $\mathcal{A}$ to $\mathcal{B}$ so that the original nodes are updated first, then the new ones in any order.

Now, every vertex of $G$ has received exactly $2d$ new neighbors, half of which are initially inactive. It is clear that the vertices of each $K_i^1$ will always stay active, since they have $3d-2$ active neighbors in $K_i^1$ itself, and at most $d+1$ other neighbors. If every vertex of $K_1^0$ stays inactive, then for $i \neq 1$, each vertex of $K_i^0$ will also do so, since they have at least $3d-1$ inactive neighbors, and at most $2d+1$ active ones. Now, every vertex of $K_1^0$ has $3d-1$ inactive and $3d-1$ active neighbors, plus the special vertex $v = v_1$. If we have $x(t)_v = 1$ for some $t > 0$, then necessarily $x(t')_w = 1$ for every vertex $w \in K_0^0$ and every $t' \geq t$. At time $t$, every vertex of $K_2^0$ then has at least $3d$ active and at most $2d$ inactive neighbors, so they will become active on the next time step, and stay that way. Inductively, every vertex of $H \setminus G$ will become active, and after that, so will the vertices of $G$, since $d$ is an upper bound for their degree in $G$. Thus the graph $H$ will eventually contain only active vertices if and only if $x(t)_v = 1$ for some time $t > 0$.
\end{proof}

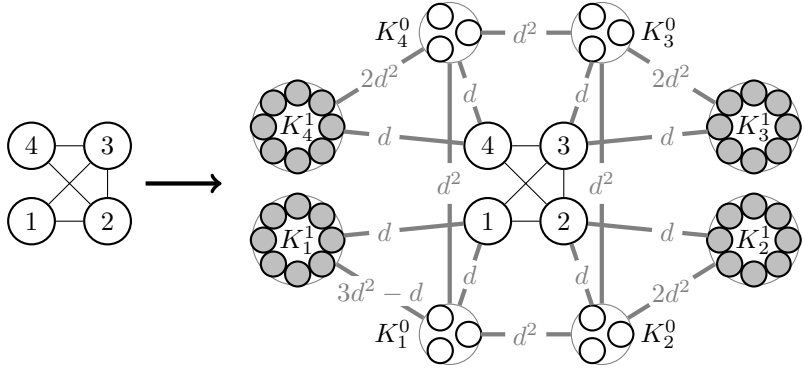
\begin{figure}[htp!]
\begin{center}
\begin{tikzpicture}

\tikzstyle{vertex}=[circle,thick,draw,fill=white,minimum size=3mm]

\node[vertex] (v1) at (0,0) {$1$};
\node[vertex] (v2) at (1,0) {$2$};
\node[vertex] (v3) at (1,1) {$3$};
\node[vertex] (v4) at (0,1) {$4$};
\draw (v2) -- (v1) -- (v3) -- (v4) -- (v2) -- (v3);

\draw[ultra thick,->] (1.5,.5) -- (2.5,.5);

\begin{scope}[xshift=6cm]

\node[vertex] (v1) at (0,0) {$1$};
\node[vertex] (v2) at (1,0) {$2$};
\node[vertex] (v3) at (1,1) {$3$};
\node[vertex] (v4) at (0,1) {$4$};
\draw (v2) -- (v1) -- (v3) -- (v4) -- (v2) -- (v3);

\foreach \x/\y [count=\i] in {-.5/-1.5,1.5/-1.5,1.5/2.5,-.5/2.5}{
	\node [draw,circle,gray,inner sep=0,minimum size=.8cm] (k\i0) at (\x,\y) {};
	\foreach \ang [count=\j] in {0,120,240}{
		\node[vertex] (k\i0\j) at ($ (\x,\y) + (\ang:.25) $) {};
	}
}

\foreach \x/\y [count=\i] in {-2.5/-.25,3.5/-.25,3.5/1.25,-2.5/1.25}{
	\node [draw,circle,gray,inner sep=0,minimum size=1.2cm] (k\i1) at (\x,\y) {};
	\foreach \ang [count=\j] in {0,45,90,135,180,225,270,315}{
		\node[vertex,fill=gray!50] (k\i1\j) at ($ (\x,\y) + (\ang:.45) $) {};
	}
}

\foreach \i/\ni in {1/2,2/3,3/4,4/1}{
	\draw[ultra thick,gray] (v\i) edge node[fill=white,circle,minimum size=0,inner sep=1] {$d$} (k\i0);
	\draw[ultra thick,gray] (v\i) edge node[fill=white,circle,minimum size=0,inner sep=1,pos=.66] {$d$} (k\i1);
	\draw[ultra thick,gray] (k\i0) edge node[fill=white,circle,minimum size=0,inner sep=1] {$d^2$} (k\ni0);
}
\draw[ultra thick,gray] (k10) edge node[fill=white,minimum size=0,inner sep=1] {$3d^2-d$} (k11);
\foreach \i in {2,3,4}{
	\draw[ultra thick,gray] (k\i0) edge node[fill=white,circle,minimum size=0,inner sep=1] {$2d^2$} (k\i1);
}

\node () at (-1.25,-1.5) {$K^0_1$};
\node () at (2.25,-1.5) {$K^0_2$};
\node () at (2.25,2.5) {$K^0_3$};
\node () at (-1.25,2.5) {$K^0_4$};

\node () at (-2.5,-.25) {$K^1_1$};
\node () at (3.5,-.25) {$K^1_2$};
\node () at (3.5,1.25) {$K^1_3$};
\node () at (-2.5,1.25) {$K^1_4$};

\end{scope}

\end{tikzpicture}
\end{center}
\label{fig:FullGadget}
\caption{The construction of Theorem~\ref{thm:Full} applied to a simple example graph of degree $d = 3$. The circular formations denote complete subgraphs, and the thick gray lines denote multiple edges (their number is given in the label).}
\end{figure}

\section{Moving the threshold}
\label{setction:proportion}

In this section we will show that our results are valid for more general update rules than the majority rule. For a real number $p \in (0,1)$ and a graph $G$, consider the rule obtained from the following function for each vertex $v \in V(G)$:
\[ f^p_i(x) = \left\{ \begin{array}{ll}
  1, & \mathrm{if~} \sum_{j\in N(i)} x_j > p\vert N(i) \vert, \\
  0, & \mathrm{if~} \sum_{j\in N(i)} x_j \leq p\vert N(i) \vert,
\end{array} \right. \]
where $N(v) \subset V(G)$ is the set of neighbors of $v$ in $G$. The rule obtained from this vertex function is called \emph{portion-$p$} rule, and the automata network with the portion-$p$ rule will be called a \emph{portion-$p$ network}. Notice that the majority network is a portion-$1/2$ network. The dynamics of a portion-$p$ network is denoted by $F_p$. For this class of functions, we can define the one cell prediction problem for a fixed $p \in (0,1)$ and a family $\SH$ of updating schemes:

\begin{framed}
\noindent {\socp$(p)$}: Given a portion-$p$ network $\mathcal{A} = (G, (f^p_v)_{v \in V})$, an updating scheme $S \in \SH$, a configuration $x \in \{0,1\}^V$ of $\mathcal{A}$, and an initially inactive vertex $v \in V(G)$, does there exist $y \in T_S(x)$ such that $y_v = 1$?
\end{framed}


Using the same arguments given for the majority network, for any portion $p \in (0,1)$ the problems \seocp$(p)$ and \syocp$(p)$ can be solved in polynomial time by simply simulating the dynamics of the network. In the following, we will show that for any $p \in (0,1)$ satisfying the appropriate computability conditions, the block sequential version \bsocp$(p)$ is \PSPACE-complete.

\begin{theorem}
For all $p \in (0,1)$ whose digits can be computed in polynomial space, the problem \bsocp$(p)$ is \PSPACE-complete.
\end{theorem}

\begin{proof}

Let $p \in (0,1)$, and assume $p < 1/2$, the other case being essentially symmetric. We will reduce \bsma{} to \bsocp$(p)$. Let $\mathcal{A}$ be a majority network, $G$ the underlying graph of $\mathcal{A}$, $S_G$ a block sequential updating scheme, $x \in \{0,1\}^V$ a configuration of $\mathcal{A}$ and $v \in V(G)$. 


For each vertex $w \in G(V)$, denote $d_G(w) = |N_G(w)|$, and let $n(w)$ be any integer satisfying
\begin{equation}
\label{eq:Proportion}
\lfloor p(d_G(w) + n(w)) \rfloor = \left\lfloor \frac{d_G(w)}{2} \right\rfloor,
\end{equation}
which exists since $p < 1/2$. We also note that $p n(w) < (1/2-p) d_G(w) + 1$ holds, so that $n(w)$ is at most linear in $d_G(w)$ when $p$ is a constant. Now, we construct a new graph $H$ from $G$ by attaching to it a complete graph $K_w$ of $\max ( n(w), \lceil 1/p \rceil + 1 )$ inactive vertices for each $w \in V(G)$, and adding a new edge from $w$ to $n(w)$ vertices of $K_w$.

We now claim that the subgraph $G$ of the portion-$p$ network $(H, (f_w)_{w \in V(H)})$ behaves identically to the majority network $\mathcal{A}$, if the new vertices are inactive, which then proves the claim. For this, it suffices to note that for all such configurations $y \in \{0,1\}^{V(H)}$ and all $w \in V(G)$ we have $\sum_{u \in N_H(w)} y_u = \sum_{u \in N_G(w)} y_u$, and by \eqref{eq:Proportion}, these integers are larger than $\frac12 d_G(w)$ if and only if they are larger than $p d_H(w)$. Thus we have $F_p(y)_w = F(y|_G)_w$, and by induction $F_p^t(y)_w = F^t(y|_G)_w$ for all $t \geq 0$. In particular, if we denote by $x' \in \{0, 1\}^{V(H)}$ the extension of $x$ to $H$, then $F^t(x)_v = 1$ holds if and only if $F_p^t(x')_v = 1$ does.
\end{proof}

We note that in the above reduction, it is essential that ties are handled the same way in majority networks and portion-$p$ networks.

\section{Conclusion}

In this article, we have studied the computational complexity of predicting the evolution of Boolean majority networks under the block sequential update schemes. It turns out that, while the sequential and synchronous schemes admit a polynomial-time prediction algorithm, predicting the behavior of a single vertex is \PSPACE-complete in the more general case of block sequential schemes. This is due to two factors: first, arbitrary Boolean circuits can be simulated by majority networks (which is also why the prediction of the sequential and synchronous schemes is \P-complete \cite{PaperGoles}), and second, the block sequential scheme allows us to construct small gadgets, which we call clocks, that empty the circuit once its computation is finished, and transport the output of the computation back into the input vertices, so that the circuit can be iterated indefinitely. We also argue that the one cell prediction problem is a fundamental one by reducing other prediction problems to it, namely, the prediction of one cell being eventually active for some time on, the prediction of all cells becoming active, and the one cell prediction problem for networks with different thresholds.

We have tried to minimize the maximum degree, block size, and the number of blocks in our theorems to some extent, but it is very likely that they can be further improved. We leave it as an open problem to find the minimum values for there parameters that still keep the problems \PSPACE-complete.

A natural continuation of this research would be to study Boolean networks with even more general update rules. For example, consider the case of \emph{AND-OR networks}, which are Boolean automata networks $\mathcal{A} = (G, (f_v)_{v \in V})$ where each vertex function $f_v$ computes either $\min N_G(v)$ or $\max N_G(v)$. In such networks, implementing a monotone Boolean circuit is already difficult, since a single inactive neighbor, causes a $\min$-vertex to become inactive in the next step, but every simulated gate should somehow be connected to the vertices that correspond to its outputs. This problem has been overcome in \cite{PaperGoles}, where it was shown that for sequential and synchronous schemes, the prediction problem of AND-OR networks is in fact \P-complete. With block sequential schemes, there again exist limit cycles of exponential lengths, which suggests that the prediction problem may be \PSPACE-complete in this case.

\bibliographystyle{plain}
\bibliography{bib}{}

\end{document}